\documentclass[11pt]{article}
\usepackage[T1]{fontenc}
\usepackage{amsfonts}
\usepackage{amsmath}
\usepackage{amssymb}
\usepackage{amsthm}
\usepackage{bbm}
\usepackage{bm}
\usepackage{mathrsfs}
\usepackage{verbatim}
\usepackage{setspace}
\usepackage{color}
\usepackage{pdfsync}
\usepackage{enumitem}
\usepackage{graphicx}
\usepackage{subfigure}
\usepackage{dsfont}

\theoremstyle{plain}
\newtheorem{theorem}{Theorem}[section]
\newtheorem{proposition}[theorem]{Proposition}
\newtheorem{lemma}[theorem]{Lemma}
\newtheorem{corollary}[theorem]{Corollary}

\theoremstyle{definition}
\newtheorem{definition}[theorem]{Definition}
\newtheorem{remark}[theorem]{Remark}

\newtheorem{example}[theorem]{Example}

\theoremstyle{remark}

\renewenvironment{thebibliography}[1]{%
\begin{oldthebibliography}{#1}%
\setlength{\baselineskip}{1em}
\linespread{.2}
\small
\setlength{\parskip}{0.25ex}%
\setlength{\itemsep}{.20em}%
}%
{%
\end{oldthebibliography}%
}
\newcommand{\eps}{\varepsilon}

\newcommand{\N}{\mathbb{N}}

\newcommand{\R}{\mathbb{R}}

\newcommand{\x}{\mathsf{x}}
\newcommand{\y}{\mathsf{y}}

\newcommand{\cC}{\mathcal{C}}

\newcommand{\cE}{\mathcal{E}}
\newcommand{\cF}{\mathcal{F}}

\DeclareMathOperator{\card}{card}

\newcommand{\qforq}{\quad\mbox{for}\quad}

\newcommand{\mykill}[1]{}

\makeatletter
\renewcommand*\env@matrix[1][\arraystretch]{%
  \edef\arraystretch{#1}%
  \hskip -\arraycolsep
  \let\@ifnextchar\new@ifnextchar
  \array{*\c@MaxMatrixCols c}}
\makeatother

\numberwithin{equation}{section}

\usepackage[pdfborder={0 0 0}]{hyperref}
\hypersetup{
  urlcolor = black,
  pdfauthor = {Marcel Nutz, Johannes Wiesel, Long Zhao},
  pdfkeywords = {Semistatic Trading, Borwein--Lewis theorem},
  pdftitle = {Limits of Semistatic Trading Strategies},
  pdfsubject = {Limits of Semistatic Trading Strategies},
  pdfpagemode = UseNone
}

\begin{document}

\title{\vspace{-1em}
  Limits of Semistatic Trading Strategies
 }
\date{\today}

\author{
  Marcel Nutz%
  \thanks{
  Departments of Statistics and Mathematics, Columbia University, mnutz@columbia.edu. Research supported by an Alfred P.\ Sloan Fellowship and NSF Grants DMS-1812661, DMS-2106056.}
  \and
  Johannes Wiesel%
  \thanks{Department of Statistics, Columbia University, johannes.wiesel@columbia.edu.}
  \and
  Long Zhao%
    \thanks{Department of Statistics, Columbia University, long.zhao@columbia.edu.}
  }
\maketitle \vspace{-1.2em}

\begin{abstract}
We show that pointwise limits of semistatic trading strategies in discrete time are again semistatic strategies. The analysis is carried out in full generality for a two-period model, and under a probabilistic condition for multi-period, multi-stock models. Our result contrasts with a counterexample of Acciaio, Larsson and Schachermayer, and shows that their observation is due to a failure of integrability rather than instability of the semistatic form. Mathematically, our results relate to the decomposability of functions as studied in the context of Schr\"odinger bridges.
\end{abstract}

\vspace{.3em}

{\small
\noindent \emph{Keywords} Semistatic Trading; Borwein--Lewis theorem

\noindent \emph{AMS 2010 Subject Classification}
91G10; %
60H05; %
26B40 %
}
\vspace{.6em}

\section{Introduction}\label{se:intro}

Closedness properties of trading strategies play a pivotal role in mathematical finance. They are at the heart of the separation arguments underlying the Fundamental Theorem of Asset Pricing, the superhedging duality, the existence of optimal portfolios for utility maximization problems, and other key results (see, e.g., \cite{DelbaenSchachermayer.06,FollmerSchied.11} and the references therein). In the most classical setting, strategies refer to dynamic trading in a stock or other liquidly traded securities. Closedness refers to the limit of a sequence of final outcomes from self-financing trading again being an outcome of some admissible trading strategy, or at least being dominated by one.  %

In addition to dynamic trading in a stock $(X_{t})$, a semistatic strategy allows for buy-and-hold trading in options written on~$X$, usually European options maturing at the time horizon~$T$ of the model. The static nature of this position reflects the increased trading cost relative to the stock. By combining options such as calls with different strikes, the trader can approximate the payoff $g(X_{T})$ for an arbitrary function~$g$. A linear pricing rule for all such payoffs is equivalent to fixing the risk-neutral distribution of~$X_{T}$ \cite{BreedenLitzenberger.78}. Taking this distribution as a primitive, a large body of literature starting with~\cite{Hobson.98} developed the theory of model-free or robust finance, where the  distribution of~$X$ is taken to be unknown or partially unknown, respectively, up to a no-arbitrage condition (see~\cite{BeiglbockHenryLaborderePenkner.11,BouchardNutz.11,GalichonHenryLabordereTouzi.11,Hobson.11}, among many others).

In the classical setting where~$X$ is modeled on a given probability space, semistatic trading is equally natural but the mathematical foundations of the theory are not well developed. We might expect this setting to lie between the usual one (with a fixed reference measure but without options) and the model-free one (with options but without reference measure). Maybe surprisingly, this is not the case: it had been observed for some time that the standard arguments for closure and separation do not apply in a straightforward way. The explanation is provided by~\cite{AcciaioLarssonSchachermayer.17} whose main result consists of two counterexamples, one in discrete and one in continuous time, stating that the space of (final outcomes from) semistatic trading strategies is not closed in the sense defined there. This clearly creates an obstacle to developing the theory along the usual lines. In this paper, we focus on the discrete-time setting and provide \emph{positive} closedness results. In particular, we illuminate what goes wrong in the example of~\cite{AcciaioLarssonSchachermayer.17}. Our results open the door to developing other aspects of mathematical finance in the semistatic setting. One such aspect, the existence of an optimal portfolio for the exponential utility maximization problem, is treated in a companion paper~\cite{NutzWieselZhao.22b} which takes our result as its starting point. A related question comes up in~\cite{Guyon.20, Guyon.21}, where trading also involves the volatility index~VIX and it is postulated that a certain optimal log-density has the form of a semistatic portfolio. The arguments of the present work should be useful to deduce this from first principles.

The aforementioned discrete-time example of~\cite{AcciaioLarssonSchachermayer.17} considers a two-period model of dynamic trading in a stock $(X_{t})_{t=0,1,2}$ and static trading in options $g(X_{2})$, where $g$ is integrable under the law $\mu$~of $X_{2}$. The physical measure~$P$ of the model is taken to be a risk-neutral measure, thus excluding issues related to arbitrage. The authors exhibit a sequence of bounded strategies with nonnegative outcomes converging in $L^{p}$ for any finite $p$ and prove that the limit of the outcomes is not the outcome of an admissible strategy (not even dominated by one). The proof is based on a clever contradiction argument which circumvents studying the limiting random variable in detail.

Our first main result below focuses on a two-period model and a sequence of semistatic strategies converging pointwise. In a non-probabilistic setting, it shows that the limiting outcome is again a semistatic strategy. This result implies stability under almost-sure convergence in a probabilistic setting by application to a set of measure one. As there are no restrictions on the probability measure, the result holds regardless of arbitrage considerations or other probabilistic assumptions. Returning to the example of~\cite{AcciaioLarssonSchachermayer.17}, our result indicates that the main failure relates to the integrability of the limiting option position rather than the semistatic functional form: the limit is still a sum of dynamic trading and an option; however, the integrability of the option can fail. (This is by no means a negligible issue---it removes the obvious way of assigning a price to the option.) Of course, one may hope that \emph{specific} limiting portfolios nevertheless enjoy good integrability properties. In the companion paper~\cite{NutzWieselZhao.22b}, we prove this for the optimal portfolio of exponential utility maximization.

The proof of our two-period, single stock result involves analyzing some finer algebraic structures of the set where convergence takes place. While our experience suggests that the result may extend to more general models, the complexity of our analysis grows rather quickly, possibly to the extent of becoming infeasible. In our second main result, we aim to exhibit a reasonably weak financial/probabilistic condition that eliminates some of the subtleties and enables a general closedness result by a more generic analysis. Indeed, our result covers the standard setting with any (finite) number of stocks and periods, and options on all individual stocks. The analysis immediately extends to trading constraints such as no-shorting, and with modest efforts it could be adapted for options at intermediate dates, options only on some of the stocks, or similar scenarios. The proposed financial condition is that the reference measure~$P$ of the model be equivalent to (i.e., have the same nullsets as) a product measure; namely, the product of all the marginals of the individual stocks at the individual dates. Intuitively, this means that none of the (a priori possible) future stock prices becomes completely impossible given an intermediate state of the prices. We remark that this interpretation is similar to the ``conditional full support'' assumption known in the theory of transaction costs; see \cite{GuasoniRasonyiSchachermayer.08} and the literature thereafter.

Remarkably, we have not been able to use the classical arguments of mathematical finance, nor the techniques known from the dual problem of martingale optimal transport~\cite{BeiglbockNutzTouzi.15} in this work, despite the setting lying between those two. Instead, we draw inspiration from the literature on the dual of the Schr\"odinger bridge problem, especially \cite{BorweinLewis.92,RuschendorfThomsen.97}. See also \cite{Leonard.14,Nutz.20} for general introductions. We shall see that in the present context, the analysis is significantly more involved as the increments of the stock create an interaction between the variables. Nevertheless, notions such as the connectedness defined in~\cite{BorweinLewis.92} are crucially helpful. The observation that semistatic trading is related to a martingale version of the Schr\"odinger bridge problem goes back to~\cite{HenryLabordere.19} which discusses a continuous-time problem.

The remainder of this paper is organized as follows. Section~\ref{se:2period} contains a complete analysis of the two-period, single-stock model. Building on its insights, Section~\ref{se:multiperiod} proposes a probabilistic condition enabling the analysis of a multi-period, multi-stock model: we show that the problem of closedness reduces to a problem of linear algebra. The latter consists in establishing that a certain matrix has full rank, which is the content of Section~\ref{se:fullRank}.

\section{Two-Period Model}\label{se:2period}

In this section, we provide a pointwise convergence analysis for a two-period model with a single stock. In our formulation, the variable $x$ represents the stock price at date~1 and $y$ the price at date~2. Equivalently, the stock price process is given by the canonical process $(X,Y)$ on $\R^{2}$. Throughout the section, we fix $E\subset\R^{2}$, interpreted as the set of possible states for~$(x,y)$.

\begin{definition}\label{de:semistatic2period} 
A function $v:E\to\R$ is a semistatic strategy if 
\begin{equation*}
 v(x,y) = h(x) (y-x) + g(y), \qquad (x,y)\in E
\end{equation*}
for some functions $h,g:\R\to\R$. We refer to~$h$ as a stock position and~$g$ as an option position.
\end{definition} 

In general, $h$ and $g$ are not uniquely determined by~$v$, hence one should think of~$v$ as a function $E\to\R$ rather than consisting of the pair~$h,g$. We note that trading between the initial date~0 and date~1 is not represented explicitly. This entails no loss of generality: if the initial state is deterministic, an expression of the form $h_{0}(x-x_{0})+h(x) (y-x) + g(y)$ can always be rewritten as a semistatic strategy in the above sense.

\begin{theorem}\label{th:main1}
The set of semistatic strategies is closed under pointwise convergence.
\end{theorem} 

While the theorem states that a limit of semistatic strategies~$v_{n}$ is again a semistatic strategy, we emphasize that the corresponding stock and option positions do not converge in general. Clearly the theorem entails its probabilistic counterparts: if convergence holds $P$-a.s.\ under some measure~$P$, we can apply the above with~$E$ being the set of full measure where convergence holds. Even if~$P$ satisfies a no-arbitrage condition, the structure of~$E$ can be quite complicated in this context, hence the importance of leaving~$E$ general in the theorem.

The proof of Theorem~\ref{th:main1} is stated at the end of Section~\ref{se:2periodAnalysis}, after an analysis which also provides detailed insight about what happens with the stock and option positions in the passage to the limit.

\subsection{Analysis of the Two-Period Model}\label{se:2periodAnalysis}

We divide the given set $E\subset\R^{2}$ into its set $E_d := \{(x,y) \in E: x = y\}$ of diagonal points and the complement $E_{o}:=E \setminus E_d$ of off-diagonal points. Consider a semistatic strategy $v(x,y)=h(x)(y-x)+g(y)$ at a point $(x,y) \in E_{o}$, then as $y-x\neq0$, the value $h(x)$ uniquely determines $g(y)$, and vice versa. Similarly, for a sequence $v_{n}(x,y)=h_{n}(x)(y-x)+g_{n}(y)$ where $v_{n}(x,y)$ converges, the value $h_{n}(x)$ converges if and only if $g_{n}(y)$ does. For $(x,y) \in E_{d}$, the situation is quite different: $v(x,y)=h(x)(y-x)+g(y)=g(y)$, so that the option position is determined (or convergent, respectively), whereas $h(x)(y-x)=0$ irrespectively of the stock position~$h(x)$ at~$x$. This does not mean that~$h(x)$ can be ignored---if $(x,y')\in E$ for some $y'\neq y$, the value~$h(x)$ is nevertheless relevant for the strategy. 

We first focus our analysis on~$E_{o}$. The following notion was introduced by~\cite{BorweinLewis.92} in a different context.

\begin{definition}\label{de:connected}
	Two points $(x,y), (x',y') \in E_{o}$ are \emph{connected}, denoted $(x,y) \sim (x',y')$, if there exist $k\in \N_{0}$ and $(x_i,y_{i})_{i=1}^{k}\in E_{o}^{k}$ such that the points
\begin{equation}\label{eq:path}
	(x,y), (x_1, y), (x_1, y_1), (x_2, y_1), \dots, (x_{k},y_{k}), (x', y_{k}), (x',y')
\end{equation}
all belong to~$E_{o}$. In that case, $(x_i,y_{i})_{i=1}^{k}$ is called a \emph{path} (in~$E_{o}$) from $(x,y)$ to $(x',y')$.
A set $C\subset E_{o}$ is connected (in~$E_{o}$) if any two points in~$C$ are connected. 
\end{definition}

For the list~\eqref{eq:path}, the crucial property is that only one coordinate is changed in each step. In our notation, the first coordinate changes first, but because a point can be repeated in the list, this entails no loss of generality.
We observe that $\sim$ is an equivalence relation on $E_{o}$. The corresponding equivalence classes $\cC=\{ C_{\gamma}: \gamma \in \Gamma \}$ are called the connected components of~$E_{o}$.

We say that \emph{uniqueness of portfolio positions} holds at $(x,y)\in E$ if for any semistatic strategy~$v: E\to\R$, the stock and option positions $h(x)$ and $g(y)$ are uniquely determined at~$(x,y)$. In that situation,  \emph{convergence of portfolio positions} holds at $(x,y)\in E$ if for any semistatic strategies~$v_{n}$ converging pointwise, the positions $h_{n}(x)$ and $g_{n}(y)$ are also convergent.

\begin{lemma}\label{le:conn}
	Let $C \subseteq E_o$ be connected. If uniqueness (convergence) of portfolio positions holds at some point of~$C$, then uniqueness (convergence) of portfolio positions holds at all points of~$C$.
\end{lemma}	

\begin{proof}
  Let uniqueness of portfolio positions hold at~$(x,y)$ and let $(x,y)\sim (x',y')$. Consider a path as in~\eqref{eq:path}, then as discussed at the beginning of Section~\ref{se:2periodAnalysis}, uniqueness of the option position at~$y$ implies uniqueness of the stock position at~$x_{1}$ which in turn implies uniqueness of the option position at~$y_{1}$, and so on, leading to uniqueness at~$(x,y)$. Similarly for the convergence.
\end{proof} 

\begin{definition}\label{de:cycle}
  A path $(x_i,y_{i})_{i=1}^{k}$ from $(x,y)\in E_{o}$ to itself %
  is called a \emph{cycle}. The cycle is \emph{identifying} if 
	\begin{equation} \label{eq:identifying}
		\prod_{i=1}^{k}(y_i - x_i) - \prod_{i=1}^{k} (y_i - x_{i+1}) \neq 0,
	\end{equation}
	where we use the cyclical convention $x_{k+1} := x_1$.
\end{definition}  

The terminology is explained by the subsequent lemma: such a cycle uniquely ``identifies'' the portfolio positions along its points. We note that $k\geq2$ must hold for any identifying cycle. If $((x_1,y_{1}),(x_{2},y_{2}))$ is a cycle, then it is identifying if and only if $x_1 \neq x_2$ and $y_1 \neq y_2$, because~\eqref{eq:identifying} reduces to $(x_1 - x_2)(y_1 - y_2) \neq 0$. In particular, the cycle can be envisioned as a nondegenerate rectangle $\{x_1, x_2\} \times \{y_1, y_2\} \subset E_{o}$. This simple characterization does not extend to larger cycles.

\begin{lemma}\label{le:identifying}
	Let $(x_i,y_{i})_{i=1}^{k}$ be an identifying cycle. Then uniqueness and convergence of portfolio positions hold at $(x_1,y_{1})$.
\end{lemma}

\begin{proof}
  Consider a semistatic strategy $v(x,y)=h(x)(y-x)+ g(y)$ along the points~\eqref{eq:path},%
  $$
  v(x,y)=h(x)(y-x)+ g(y) \qforq (x,y)\in \{(x_i,y_{i}),(x_{i+1},y_{i}):\, 1\leq i\leq k\}. 
  $$
  This can be cast as the $2k\times2k$ linear system 
  $$
  	\begin{bmatrix}
	y_1 - x_1 & 1  \\
	               & 1 & y_1- x_2 \\
			  &    & y_2 - x_2 & 1 \\
			  &    &                & 1 & y_2 - x_3 \\
	               &    &                &    &     1        & \ddots  \\
			  &    &                &    &               & \ddots  & \ddots \\
	               &    &                &    &               &           & y_k - x_k & 1 \\
	 y_k - x_1 &    &                &    &               &           &               & 1
	\end{bmatrix} \hspace{-.1em}\begin{bmatrix}
	h(x_1) \\ g(y_1) \\ h(x_2) \\ g(y_2) \\ \vdots \\ h(x_k) \\ g(y_k) \end{bmatrix} \hspace{-.1em}=\hspace{-.1em}
	\begin{bmatrix}
		v(x_1, y_1) \\ v(x_2, y_1) \\ v(x_2, y_2) \\ v(x_3, y_2) \\ \vdots \\ v(x_k, y_k) \\ v(x_1, y_k)
	\end{bmatrix}
  $$
  where omitted matrix entries are zero.
  Using Laplace expansion along the first column and the convention $x_{k+1} := x_1$, we see that the determinant of the matrix is 
  \begin{align*}
     (y_1-x_1) &\left[ (y_2-x_2) \cdots (y_k-x_k) \right] \\
    + (-1)^{2k+1}(y_k-x_1) &\left[(y_1 - x_2) \cdots (y_{k-1}-x_k)\right] \\
    & = \prod_{i=1}^{k}(y_i - x_i) - \prod_{i=1}^{k} (y_i - x_{i+1}).
  \end{align*}
  As the cycle is identifying, it follows that the matrix is invertible, and the inverse map is continuous as a finite-dimensional linear map. In summary, the numbers $(h(x_{i}),g(y_{i}))_{1\leq i\leq k}$ are uniquely determined by a continuous function of the numbers $(v(x_{i},y_{i}), v(x_{i+1},y_{i}))_{1\leq i\leq k}$, showing the result.
\end{proof}

Combining Lemma~\ref{le:conn} and Lemma~\ref{le:identifying}, we have the following.

\begin{corollary}\label{co:identifying}
  Let $C \subseteq E_o$ be connected. If $C$ contains an identifying cycle, then uniqueness and convergence of portfolio positions hold on~$C$.
\end{corollary} 

Next, we study what happens in the absence of identifying cycles. Given a subset $S\subset\R^{2}$, we denote by~$S^{\x}$ and~$S^{\y}$ its projections onto the first and second coordinate, respectively.

\begin{proposition}\label{pr:noIdentif} 
	Let $C$ be a connected component of $E_{o}$ containing no identifying cycles. 
	\begin{enumerate}
	\item[(a)] 
	Uniqueness of portfolio positions fails at each point of~$C$. For any semistatic strategy, the set of portfolio positions is a one-parameter family. 
	\item[(b)] 
	Closedness of semistatic strategies holds on~$C$. More precisely, let $v_n(x,y) =  h_n(x) (y-x) + g_n(y)$ be semistatic strategies converging pointwise on~$C$ to $v: C \to \R$. Then there exist
	$h'_{n},g'_{n}$ such that 
	$$
	  v_n(x,y) =  h'_n(x) (y-x) + g'_n(y), \qquad (x,y)\in C
	$$
	and their pointwise limits exist,
	$$
	  h':=\lim h'_{n} \quad\mbox{on }C^{\x}, \qquad g':=\lim g'_{n}\quad\mbox{on }C^{\y}.
	$$
	In particular, $v(x,y) =  h'(x) (y-x) + g'(y)$ on~$C$. 
	
	The positions $h'_{n},g'_{n}$ can be constructed as follows. Fix an arbitrary point $(x_0, y_0) \in C$. Then there exist unique functions $a, b: \R \to \R\setminus\{0\}$ such that $a(x) b(y) = y - x$ on $C$ and $a(x_{0})=1$. We can choose 
	$$
	  h'_n(x):=h_n(x) - \frac{h_n(x_0)}{a(x)}, \qquad   g'_n(y):=g_n(y) + h_n(x_0) b(y).
	$$
	\end{enumerate}
\end{proposition}

Before stating the proof, we recall the Borwein--Lewis characterization for the decomposability of a function of two variables into a product of single-variable functions. More generally, this result applies to group-valued functions on arbitrary sets; see~\cite[Theorem~3.3]{BorweinLewis.92}. %

\begin{lemma}\label{le:BorweinLewis}
  Let $S\subset\R\times \R$ and $c: S \to \R \setminus \{0\}$. The following are equivalent:
	\begin{enumerate}
		\item There exist $a,b: \R\to\R\setminus \{0\}$ such that
		$$
		c(x,y) = a(x) b(y), \qquad (x,y)\in S.
		$$
	  \item For any cycle $(x_i,y_{i})_{i=1}^{k}$ in $S$,
      \begin{equation} \label{eq: Borwein-Lewis cond}
	      \prod_{j=1}^{k} c(x_j, y_j) = \prod_{j=1}^{k} c(x_{j+1}, y_j),
      \end{equation}
where $x_{k+1} := x_1$.
	\end{enumerate}
  In that case, on each connected component of $S$, the functions $a$ and $b$ are unique up on a scalar multiple.
\end{lemma}

\begin{proof}[Proof of Proposition~\ref{pr:noIdentif}.]
	The absence of identifying cycles means that~\eqref{eq: Borwein-Lewis cond} holds for the function $c(x,y):=y-x$ on~$C$. Note that~$c$ is valued in~$\R\setminus \{0\}$ due to $C\subset E_{o}$. Hence Lemma~\ref{le:BorweinLewis} implies that there exist $a, b: \R \to \R\setminus\{0\}$ such that $a(x) b(y) = y - x$ on~$C$, and these functions are uniquely determined by the normalization that $a(x_{0})=1$ for some fixed $x_0\in C^{\x}$.
	
 Consider a semistatic strategy $v(x,y) = h(x)(y-x) + g(y)$ on~$C$. Given $\alpha\in\R$, let $h_\alpha(x):=h(x)+\alpha/a(x)$ and $g_\alpha(y):=g(y)-\alpha b(y)$. Then
  \begin{align*}
	 h_\alpha(x) (y-x) + g_\alpha(y) 
	 & = [h(x)+\alpha/a(x)] a(x) b(y) + g(y)-\alpha b(y) \\
	 & = h(x) a(x)b(y) + g(y) = v(x,y),
  \end{align*} 	
showing that the portfolio positions inducing~$v$ include the one-parameter family $(h_\alpha,g_\alpha)_{\alpha\in\R}$.
	Conversely, by connectedness, we know that portfolio positions are uniquely determined as soon as the option position is determined at one point~$y_{0}$. Because $\alpha\mapsto g_\alpha(y_{0})=g(y_{0})-\alpha b(y_{0})$ is surjective onto~$\R$, this shows that $(h_\alpha,g_\alpha)_{\alpha\in\R}$ exhausts all portfolio positions inducing~$v$.
	
	Turning to the convergence, note that 
\begin{equation*}
	\bar{v}_{n}(x,y):=\frac{v_n(x,y)}{b(y)} = h_n(x) a(x) + \frac{g_n(y)}{b(y)}
\end{equation*}
on~$C$. Define
	$$
		 \bar{h}_n(x) :=h_n(x) a(x) - h_n(x_0) a(x_0), \qquad \bar{g}_n(y) := \frac{g_n(y)}{b(y)} + h_n(x_0) a(x_0),
	$$
so that 
  \begin{equation}\label{eq:additiveDecomp}
  \bar{v}_{n}(x,y) = \bar{h}_n(x) + \bar{g}_n(y).
  \end{equation}
Clearly $\bar{v}_n(x,y)$ is convergent for all $(x,y)\in C$. Moreover, $\bar{h}_n(x_0) = 0$ for all $n$, so that $\bar{h}_n(x_0)$ is convergent. As~$C$ is connected, the additive decomposition~\eqref{eq:additiveDecomp} implies as in the proof of Lemma~\ref{le:conn} (or  \cite{BorweinLewis.92}) that the separate limits $\bar{g}(y):= \lim_n \bar{g}_n(y)$  and $\bar{h}(x):= \lim_n \bar{h}_n(x)$ exist for all $(x,y)\in C$. 
It follows that 
\begin{align*}
  h'_n(x)&=\bar{h}_{n}(x)/a(x) \;\to\; \bar{h}(x)/a(x)=h'(x), \\
  g'_n(y)&=\bar{g}_{n}(y)b(y) \;\to\; \bar{g}(y)b(y)=g'(y),
\end{align*}  
completing the proof.
\end{proof}

We can now prove the main result.

\begin{proof}[Proof of Theorem \ref{th:main1}]
  Let $v_n(x,y) =  h_n(x) (y-x) + g_n(y)$ be semistatic strategies converging pointwise on~$E$ to $v: E \to \R$. We shall construct $h: E^{\x}\to \R$ and $g: E^{\y}\to \R$ such that $v(x,y) = h(x) (y-x) + g(y)$ on $E$. 

  Recall the partition $(C_{\gamma})_{\gamma\in\Gamma}$ of $E_{o}$ into connected components. The definition of $\sim$ implies that for each $\gamma\in\Gamma$,
  $$
    C_{\gamma} = \big(C_{\gamma}^{\x} \times C_{\gamma}^{\y} \big)\cap E_{o} \subset \big(C_{\gamma}^{\x} \times C_{\gamma}^{\y} \big)\cap E;
  $$
  the last inclusion can be strict as $C_{\gamma}^{\x} \times C_{\gamma}^{\y}$ can contain points from the diagonal~$E_{d}$. Conversely, some points from the diagonal may not pertain to~$C_{\gamma}^{\x} \times C_{\gamma}^{\y}$ for any~$\gamma\in\Gamma$; these points form the set
  $$
    N := E\setminus \bigcup_{\gamma\in\Gamma}\big(C_{\gamma}^{\x} \times C_{\gamma}^{\y} \big) \subseteq E_d.
  $$
  
  Let $(D_{j})_{j\in J}$ be the collection consisting of all singletons $\{(x,y)\}$ with $(x,y)\in N$ as well as the sets $\big(C_{\gamma}^{\x} \times C_{\gamma}^{\y} \big)\cap E$, $\gamma\in\Gamma$. Note that two points in $E_{o}$ sharing one coordinate are necessarily connected, and two points in $E_{d}$ sharing one coordinate must coincide. This implies that $(D_{j})_{j\in J}$ is a partition of~$E$ with the following property: If $(x,y),(x,y')\in E$, then $(x,y)$ and $(x,y')$ belong to the same component~$D_{j}$. Similarly, if $(x,y),(x',y)\in E$, then $(x,y)$ and $(x',y)$ belong to the same component~$D_{j}$. As a consequence, we may construct the positions $h,g$ separately on each~$D_{j}$ without danger of creating any inconsistencies. (This is not true for $\{C_{\gamma},E_{d}\}$, whence the need for yet another collection.)
  \begin{enumerate}
   \item
   Let $D_{j} = (C_{\gamma}^{\x} \times C_{\gamma}^{\y})\cap E$ for some $\gamma\in\Gamma$, where $C_{\gamma}$ contains an identifying cycle.  Then we can choose $(h,g):=\lim (h_{n},g_{n})$ on $(C_{\gamma}^{\x}, C_{\gamma}^{\y})$ according to Corollary~\ref{co:identifying}.
   \item
   Let $D_{j} = (C_{\gamma}^{\x} \times C_{\gamma}^{\y})\cap E$ for some $\gamma\in\Gamma$, where $C_{\gamma}$ does not contain an identifying cycle. Then we can choose $(h,g)$ on $(C_{\gamma}^{\x}, C_{\gamma}^{\y})$ according to Proposition~\ref{pr:noIdentif}.
   \item 
   Let $D_{j} = \{(x,y)\}$ for some $(x,y)\in N \subseteq E_d$. Then we can define $h(x)=0$ and $g(y)=v(x,y)$. In fact, as $y-x=0$, any choice for $h(x)$ will do. \qedhere
  \end{enumerate}  
\end{proof}

The preceding analysis also quantifies the non-uniqueness for portfolio positions; that is, the exact number of degrees of freedom in choosing the portfolio positions for any given semistatic strategy.

\begin{remark}\label{rk:degreesOfFreedom}
  Given any semistatic strategy~$v$ on~$E$, the set of all portfolio positions $(h,g)$ with $v(x,y) = h(x) (y-x) + g(y)$ for all $(x,y)\in E$ is a $k$-parameter family, where
  $$
  k = \card \big(N\cup\{\gamma\in\Gamma:\, C_{\gamma} \mbox{ contains no identifying cycle}\}\big).
  $$
\end{remark}

\section{General Probabilistic Model}\label{se:multiperiod}

Given integers $d \geq 1$ and $T \geq 2$, we denote by $X = (X_t)_{t=1}^T$ the canonical process on $(\R^d)^T$, where $X_t = (X_{t,j})_{j=1}^d$ are interpreted as the prices of~$d$ stocks at date~$t$. We also fix a probability measure~$P$ on~$(\R^d)^T$; only the nullsets of~$P$ will matter for our results. For the purposes of this section, a \emph{semistatic strategy} is a random variable~$V$ satisfying
\begin{equation} \label{eq: multiperiod v}
	V =
	\sum_{t=1}^{T-1} \sum_{j=1}^{d} \hat h_{t, j}(X_1,\dots,X_t) \,(X_{t+1, j} - X_{t, j}) + \sum_{j=1}^d \hat g_j(X_{T, j}) 
\end{equation}
 $P$-a.s.\ for some real-valued measurable functions~$\hat h_{t,j}$ and~$\hat g_{j}$. It will be notationally convenient to work instead with the random variables
$$
  h_{t, j}:=\hat h_{t, j}(X_1,\dots,X_t), \qquad g_{j}:=\hat g_j(X_{T, j}).
$$
We call $h=(h_{t,j})$ the stock position and $g=(g_{j})$ the option position, respectively. Together, they form the portfolio position~$(h,g)$ of~$V$. The portfolio position is not uniquely determined by~$V$ in general: it is clearly possible to add a constant to $g_{1}$ and subtract the same from $g_{j}$ for any~$j\neq1$, without affecting~$V$. These $d-1$ degrees of freedom are easily removed by fixing an ``anchor'' point $x^0 \in (\R^d)^T$ and normalizing
\begin{equation} \label{eq: initial cond}
	g_j ( x^0 ) = 0, \quad j = 2, \dots, d.
\end{equation}
In the context of Theorem~\ref{thm: multiperiod} below, it will be shown that the anchor point can be chosen arbitrarily outside a certain nullset. 

In this probabilistic setting, we say that \emph{uniqueness of portfolio positions} holds if for any semistatic strategy $V$, after a normalization of the form~\eqref{eq: initial cond}, the portfolio position $(h, g)$ is uniquely determined $P$-a.s. \emph{Convergence of portfolio positions} holds if, after a normalization of the form~\eqref{eq: initial cond}, for any semistatic strategies $V^{(n)}$ converging $P$-a.s., the corresponding  positions $(h^{(n)}, g^{(n)})$ also converge $P$-a.s.

The aim of this section  is to exhibit a probabilistic condition circumventing some of the complications highlighted in Section~\ref{se:2period}. To that end, let $\mu_{t,j}$ be the law of $X_{t,j}$, or equivalently, the one-dimensional marginal law of~$P$ on the component~$(t,j)$. We assume throughout that~$\mu_{t,i}$ is not a Dirac measure for any $t,j$. (This serves to simplify the exposition; while the results and arguments could be generalized, the degenerate case is not relevant financially and hence omitted.) The key condition for our result is that~$P$ be measure-theoretically equivalent to the product of its marginals.

\begin{theorem} \label{thm: multiperiod}
  Suppose that $P \sim \otimes_{t=1}^T \otimes_{j=1}^d \mu_{t,j}$. 
	Then the set of semistatic strategies is closed under $P$-a.s.\ convergence. 
Moreover, uniqueness and convergence of portfolio positions hold.
\end{theorem}

As mentioned in the Introduction, $P \sim \otimes_{t=1}^T \otimes_{j=1}^d \mu_{t,j}$ intuitively means that none of the (a priori possible) future stock prices becomes completely impossible given an intermediate state of the prices, similarly as in the condition of conditional full support~\cite{GuasoniRasonyiSchachermayer.08}. Technically, we shall see that for such~$P$, any set of full measure contains an abundance of cuboids that will play the role of identifying cycles (cf.\ Section~\ref{se:2period}). Indeed, fix $x^0, x^1 \in (\R^d)^T$ such that $x_{t,j}^0 \neq x_{t,j}^1$ for all $t, j$ and consider the cuboid~$D$ generated by their components, 
\begin{equation} \label{eq: D}
	D := \prod_{t=1}^T \prod_{j=1}^d \left\{ x_{t, j}^0, x_{t, j}^1 \right\} \subseteq (\R^d)^T.
\end{equation}
That is, each point in~$D$ is a matrix $(x_{t, j}^{\eps_{t,j}})_{t,j}$ where $\eps_{t,j}\in\{0,1\}$.

\begin{proposition} \label{prop: multiperiod D}
Let $V$ and $V^{(n)}$ be semistatic strategies with portfolio positions $(h,g)$ and $(h^{(n)}, g^{(n)})$, respectively, such that~\eqref{eq: multiperiod v} and~\eqref{eq: initial cond} hold on~$D$. Then $(h,g)$ are uniquely determined on $D$. Moreover, $V^{(n)} \to V$ pointwise on $D$ implies that $(h^{(n)}, g^{(n)}) \to (h,g)$ pointwise on $D$.
\end{proposition}

The proof is lengthy and deferred to Section~\ref{se:fullRank}. In a nutshell, we view~\eqref{eq: multiperiod v} as linear system where the values of $h$ and $g$ at points in $D$ are the variables; each equation of the system corresponds to evaluating $V$ at a point in $D$. We prove that the (finite-dimensional) linear map associated with the system is injective, hence admits a continuous inverse. As a result, the portfolio position $(h,g)$ is a continuous function of the strategy $V$.

The next lemma is a general measure-theoretic fact; it formalizes the claim that there is an abundance of cuboids in any set of full $P$-measure.
\begin{lemma} \label{lemma: multiperiod anchor}
	Consider probability spaces $(\Omega_i, \cF_i, \mu_i)_{i=1}^n$ and their product $(\Omega, \cF, \mu)$ given by $\Omega = \prod_{i=1}^{n}\Omega_i, \ \cF = \otimes_{i=1}^{n} \cF_i$ and $\mu = \otimes_{i=1}^{n} \mu_i$. If $A \in \cF$ satisfies $\mu(A) = 1$, then $\mu$-a.e.\ $x^0 = (x_1^0, x_2^0, \dots, x_n^0) \in A$ satisfies
	$$
		\mu \left\{ x \in A: \, \prod_{i=1}^n \left\{ x_{i}^0, x_{i} \right\} \subseteq A \right\} = 1.
	$$
\end{lemma}

\begin{proof}
	Let $\Omega^0,\Omega^1$ be two copies of $\Omega$ with components denoted $\Omega_i^j$. Consider the product space $\hat{\Omega} =\prod_{i=1}^{n} (\Omega_i^0 \times \Omega_i^1)$ endowed with the product $\sigma$-field $\hat{\cF}$ and the product measure $\hat{\mu} = \otimes_{i=1}^{n} (\mu_i \otimes \mu_i)$. For each multi-index $J = (j_1, j_2, \dots, j_n) \in \{0, 1\}^n$, define the projection $\pi_{J}: \hat{\Omega} \to \prod_{i=1}^{n} \Omega_i^{j_i}$ by 
	$$
		(x_1^0, x_1^1, x_2^0, x_2^1, \dots, x_n^0, x_n^1) \longmapsto (x_1^{j_1}, x_2^{j_2}, \dots, x_n^{j_n}).
	$$
	Clearly the law of $\pi_{J}$ under $\hat{\mu}$ is~$\mu$, so that~$\hat{\mu}(\pi_{J} \in A) =\mu(A)=1$.
	Defining
	$
		S:= \cap \{ \pi_{J} \in A \}
	$
	as the intersection over all~$J \in \{0, 1\}^n$, it follows that~$\hat{\mu}(S) = 1$. Denote by $S_{x^0}$ the section of~$S$ at~$x^0 \in \Omega^0$. In view of $\hat{\mu}(S) = 1$, 
Fubini's theorem implies
	$$
		\mu \left\{ x^0 \in \Omega^0: \mu \left(S_{x^0} \right) = 1  \right\} = 1.
	$$
 The desired result follows once we observe that $(x_1^0, x_1, x_2^0, x_2, \dots,  x_n^0, x_n) \in S$ if and only if $\prod_{i=1}^n \left\{ x_{i}^0, x_{i} \right\} \subseteq A$.
\end{proof}

We can now deduce the main result.

\begin{proof}[Proof of Theorem \ref{thm: multiperiod}]
	Suppose that $V^{(n)}$ is of the form \eqref{eq: multiperiod v} with portfolio position $(h^{(n)}, g^{(n)})$ and that $V^{(n)} \to V$ on $A \subseteq (\R^d)^T$ with $P(A) = 1$. In view of  $P \sim \otimes_{t=1}^T \otimes_{j=1}^d \mu_{t,j}$, Lemma~\ref{lemma: multiperiod anchor} implies that there exists $x^0 \in A$ such that
$$
	B:= \left\{ x \in A: \, \prod_{t=1}^T \prod_{j=1}^d \left\{ x_{t, j}^0, x_{t, j} \right\} \subseteq A \right\}
$$
satisfies $P(B)=1$.  (In fact, $P$-almost any $x^0 \in A$ will do.) We use this point~$x^{0}$ for the normalization~\eqref{eq: initial cond}.

We claim that the limit $(h(x), g(x)) := \lim_n (h^{(n)}(x), g^{(n)}(x))$ exists for all~$x\in B$. To prove this, let $x^1 \in B$ satisfy $x_{t,j}^1 \neq x_{t,j}^0$ for all~$t,j$ and consider the cuboid~$D$ determined by $x^{0}$ and~$x^{1}$; cf.~\eqref{eq: D}. Then we see from Proposition~\ref{prop: multiperiod D} that $(h^{(n)}, g^{(n)})$ converges to some $(h, g)$ on $D$, and that $(h,g)$ is uniquely determined on $D$. In particular, the limit exists at $x:=x^{0}$ and at $x:=x^{1}$. Given an arbitrary $x\in B$, as $\mu_{t,j}$ is not a Dirac measure, we can find $x^1 \in B$ such that $x_{t,j}^1 = x_{t,j}$ if $x_{t,j} \neq x_{t,j}^0$ and $x_{t,j}^1 \neq x_{t,j}^0$ for all~$t,j$. Applying the above to~$x^{0}$ and $x^{1}$, we see that the limit exists at~$x$. The same argument also establishes the uniqueness and convergence of portfolio positions.
\end{proof}

\section{Proof of Proposition \ref{prop: multiperiod D}} \label{se:fullRank}

Suppose that $V: D \to \R$ is of the form \eqref{eq: multiperiod v} with portfolio position $(h, g)$, where $g$ satisfies (\ref{eq: initial cond}) at $x^0$. 
This induces a linear system with the values of~$h$ and $g$ at the points in~$D$ as variables and the price increments $\Delta_{t+1, j}=X_{t+1, j} - X_{t, j}$ as coefficients. To start with the simplest example, consider $T=2$ and $d=1$, so that
$
	V = h_{1,1} \,(X_{2,1} - X_{1,1}) + g_1
$
and
$
	D = \left\{ (x_{1,1}^0, x_{2,1}^0), (x_{1,1}^0, x_{2,1}^1), (x_{1,1}^1, x_{2,1}^0), (x_{1,1}^1, x_{2,1}^1) \right\}.
$
This corresponds to 4 equations and can be cast as a $4 \times 4$ linear system
\begin{equation} \label{eq: T=2 d=1}
	\begin{bmatrix}[1.33]
		x_{2,1}^0 - x_{1,1}^0 & & 1 & \\
		x_{2,1}^1 - x_{1,1}^0 & & & 1 \\
		& x_{2,1}^0 - x_{1,1}^1 & 1 & \\
		& x_{2,1}^1 - x_{1,1}^1 & & 1 \end{bmatrix} \begin{bmatrix}[1.33]
		\hat h_{1,1} (x_{1,1}^0) \\ \hat h_{1,1} (x_{1,1}^1) \\ \hat g_1(x_{2,1}^0) \\ \hat g_1(x_{2,1}^1) \end{bmatrix} = \begin{bmatrix}[1.33]
		V(x_{1,1}^0, x_{2,1}^0) \\ V(x_{1,1}^0, x_{2,1}^1) \\ V(x_{1,1}^1, x_{2,1}^0) \\ V(x_{1,1}^1, x_{2,1}^1) \end{bmatrix}.
\end{equation}
In this example, the condition~\eqref{eq: initial cond} is vacuous as $d=1$.

In the general case, \eqref{eq: multiperiod v} on $D$ can be viewed as a $N_r \times N_c$ linear system that we describe next.
For $1 \leq t \leq T$, consider the binary vector $\eps_t = (\eps_{t,1}, \eps_{t,2}, \dots, \eps_{t,d}) \in \{0,1\}^d$. We view 
$$
	\vec \eps_t := (\eps_{s})_{s=1}^t := (\eps_{1,1}, \dots, \eps_{1,d}, \eps_{2,1}, \dots, \eps_{2,d}, \dots, \eps_{t,1}, \dots, \eps_{t,d}) \in \{0,1\}^{\,d\,t}
$$ 
as a multi-index of length $d\,t$ and denote the collection of all such $\vec \eps_t$ by~$\cE_t$. 
There is a one-to-one correspondence between~$\cE_T$ and~$D$ via $\vec \eps_T \mapsto x^{\vec \eps_T} := (x_{t,j}^{\eps_{t,j}})_{t,j}$.
More generally, for~$1 \leq t \leq T$, the set $\cE_t$ corresponds to the set $\prod_{s=1}^t \prod_{j=1}^{d} \{ x_{s,j}^0, x_{s,j}^1\}$ via $\vec \eps_t \mapsto x^{\vec \eps_t} := (x_{s,j}^{\eps_{s,j}})_{s,j}$, where $1 \leq s \leq t$ and $1 \leq j \leq d$.

Every point $x^{\vec \eps_T} \in D$ gives rise to an equation as we evaluate $V$ at $x^{\vec \eps_T}$. Hence the number of rows in our system is $N_r = 2^{\, d \, T}$, the cardinality of $\cE_T$. On the other hand, for each $1 \leq t \leq T-1$ and each $1 \leq j \leq d$, $\hat h_{t,j}$ gives rise to $2^{\,d \, t}$ variables, namely $\hat h_{t,j} (x^{\vec \eps_t})$ for $\vec \eps_t \in \cE_t$. In addition, recalling~\eqref{eq: initial cond}, $\hat g$ gives rise to $d+1$ variables, namely $\hat g_1 (x_{T, 1}^0)$ and $\hat g_j (x_{T,j}^1)$ for $1 \leq j \leq d$. As a result, the total number of variables (hence columns in the system) is
$$
	N_c = d \sum_{t=1}^{T-1} 2^{\,d\,t} + (d+1) = d \, \frac{2^{\,d\,T} -  1}{2^d - 1}+ 1.
$$

As in~\eqref{eq: T=2 d=1}, the coefficients of the matrix are given by stock price increments $(x_{t+1, j}^{\eps_{t+1,j}} - x_{t, j}^{\eps_{t,j}})_{t,j}$ and binary entries related to the options.
To unambiguously determine the matrix, we need to specify an order for the rows and columns; in fact, we tailor our order to facilitate the exposition below.
For $1 \leq t \leq T$, consider the natural lexicographic order on $\cE_t$ and equip $(x^{\vec \eps_t})_{\vec \eps_t \in \cE_t}$ with the induced order;
that is, $x^{\vec \eps_t} \leq x^{\vec \eta_t}$ if and only if $\vec \eps_t \leq \vec \eta_t$.
In particular, the order on $\cE_T$ induces an order on $D = (x^{\vec \eps_T})_{\vec \eps_T \in \cE_T}$.
The row ordering is set by evaluating \eqref{eq: multiperiod v} on $D$ in that order. 
The column ordering is given by the following (top-to-bottom) hierarchy:
\begin{enumerate}
	\item The variables from $\hat h$ appear before those from $\hat g$.
	\item The $d \sum_{t=1}^{T-1} 2^{\,d\,t}$ variables $\hat h_{t,j}(x^{\vec \eps_t})$ are sorted
	\begin{enumerate}
		\item in \emph{descending} order of $t \in \{T-1, T-2, \dots, 1\}$,
		\item then in ascending order of $\vec \eps_t \in \cE_t$, and
		\item lastly in ascending order of $j \in \{1, 2, \dots, d\}$. 
	\end{enumerate}
	\item The $d+1$ variables from $\hat g$ are ordered as
	$$
		\hat g_{1}(x_{T,1}^0),\, \hat g_{1}(x_{T,1}^1),\, \hat g_{2}(x_{T,2}^1),\, \hat g_{3}(x_{T,3}^1), \,\dots\,, \,\hat g_{d}(x_{T,d}^1).
	$$
\end{enumerate}
The reader can verify that the matrix in \eqref{eq: T=2 d=1} follows the desired ordering; a more advanced example can be found in \eqref{eq: column order} below.
For the general case  $T \geq 2$ and $d \geq 1$, the above convention uniquely determines a matrix~$L_{T}$ which will be shown to have the following property.

\begin{lemma} \label{lemma: rank}
	The matrix $L_T$ has full column rank.
\end{lemma}

\begin{remark}
	In the case $d=1$ of a single stock, $L_T$ is a square matrix and
\begin{align*}
	\det L_T  =  (-1)^{T-1} \, \left[ \, \prod_{t=1}^{T-1} \left( x_{t,1}^0 - x_{t,1}^1 \right)^{2^{t-1}} \right] \, \left( x_{T,1}^0 - x_{T,1}^1 \right)^{2^{T-1} - 1} \neq 0.
\end{align*}
The proof uses arguments similar to the proof of Lemma~\ref{lemma: rank} below. %
\end{remark}

Once the lemma is established, Proposition~\ref{prop: multiperiod D} is a direct consequence:

\begin{proof}[Proof of Proposition \ref{prop: multiperiod D}]
 By Lemma~\ref{lemma: rank}, the linear map associated with $L_{T}$ is injective. Its inverse is continuous as a linear finite-dimensional map, showing that the portfolio position is a continuous function of the strategy.
\end{proof}

\subsection{Proof of Lemma~\ref{lemma: rank}}

Next, we introduce some additional notation for the proof of Lemma~\ref{lemma: rank}.
Given $1 \leq t \leq T-1$ and $\vec \eps_{t} \in \cE_{t}$, we define $D^{\vec \eps_{t}} \subseteq D$ by 
\begin{equation} \label{eq: D subset}
	D^{\vec \eps_{t}} := \left\{ x^{\vec \eps_{t}} \right\} \times \left( \prod_{s=t+1}^{T} \prod_{j=1}^d \left\{ x_{s, j}^0, x_{s, j}^1 \right\} \right).
\end{equation}
That is, $D^{\vec \eps_{t}}$ consists of the $2^{\,d\,(T-t)}$ points in $D$ that share the vector $x^{\vec \eps_{t}}$ in their first $d\,t$~coordinates. Note that for each $1 \leq t \leq T-1$, $\{ D^{\vec \eps_{t}} : \vec\eps_{t} \in \cE_{t} \}$ forms a partition of $D$.

For each $1 \leq t \leq T-1$, $1\leq j \leq d$ and $(k, l) \in \{0,1\}^2$, we introduce a shorthand for the corresponding stock price increment
\begin{equation} \label{eq: Delta def}
	\Delta_{t+1, j}^{k, l} :=  x_{t+1,j}^k - x_{t,j}^{l},
\end{equation}
as well as the stock price difference
\begin{equation} \label{eq: delta def}
	\delta_{t,j} := x_{t,j}^0 - x_{t,j}^1
\end{equation}
which is nonzero by assumption. We record two identities for later use,%
\begin{equation} \label{eq: Delta diff}
	\Delta_{t+1,j}^{0,l} - \Delta_{t+1,j}^{1,l} = \delta_{t+1, j} \quad \text{and} \quad \Delta_{t+1,j}^{l,0} - \Delta_{t+1,j}^{l,1} = -\delta_{t, j},
\end{equation}
where the right-hand sides do not depend on $l \in \{0,1\}$.

As the proof of Lemma \ref{lemma: rank} is somewhat involved, we first state an example to illustrate some of the arguments.

\begin{figure}[tbh]
		\centering
		\subfigure[$L_2$ in Example \ref{ex: T=2 d=2}]{
			\label{fig: mat}
			\includegraphics[width=0.45\textwidth, trim={1.5cm, 17cm, 12cm, 2cm}, clip]{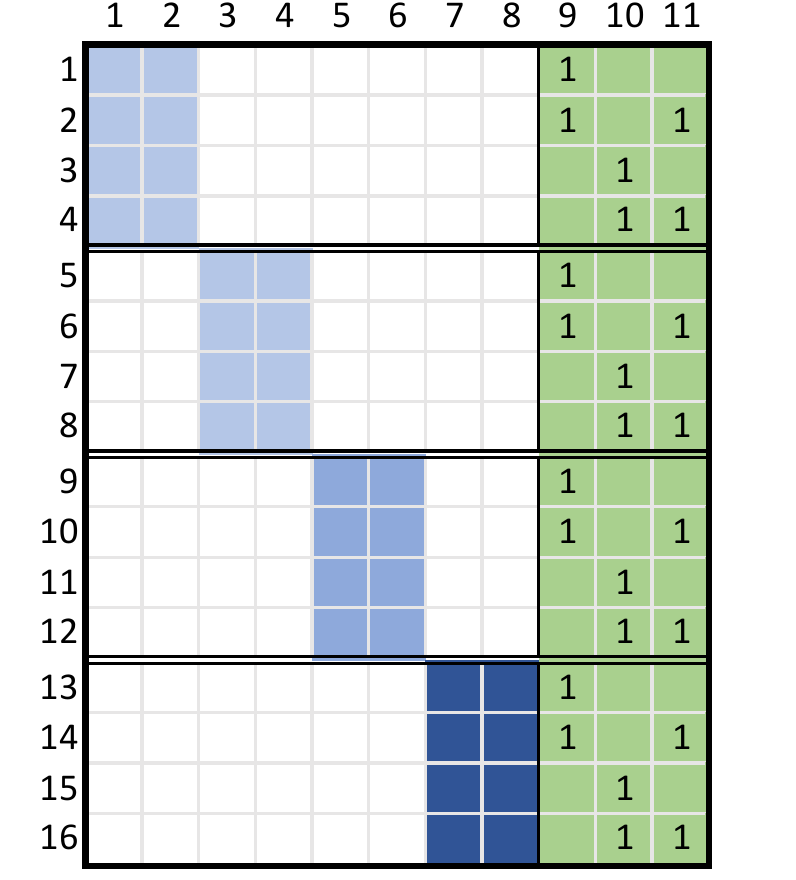}}
		\qquad
		\subfigure[$L_2'$ in Example \ref{ex: T=2 d=2}]{
			\label{fig: mat2}
			\includegraphics[width=0.45\textwidth, trim={1.5cm, 17cm, 12cm, 2cm}, clip]{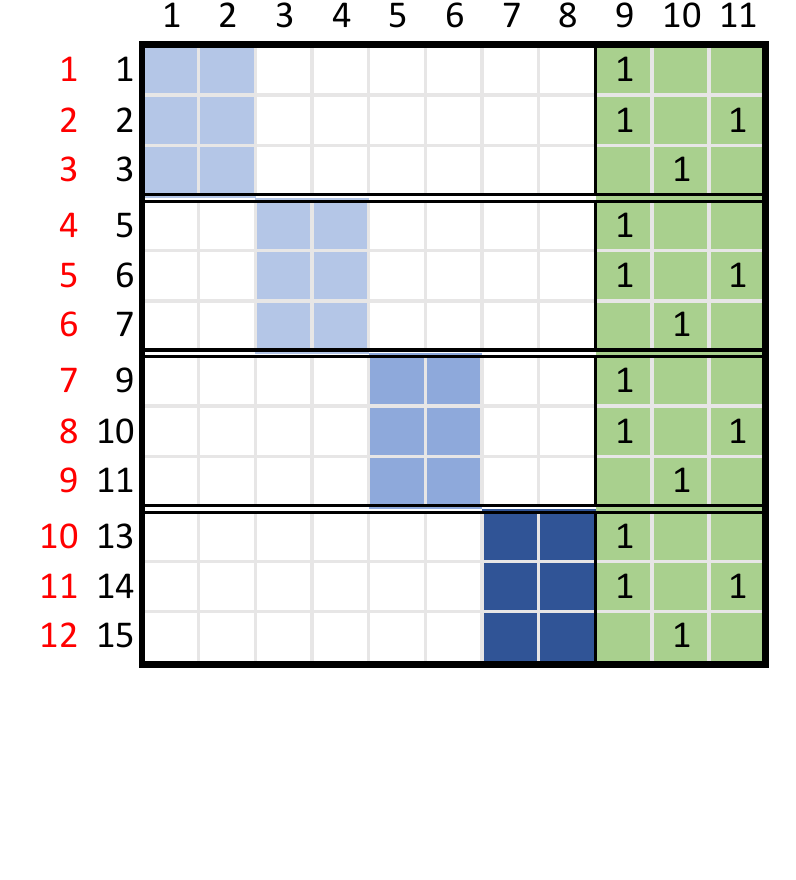}}
		\caption{
		The blue and green blocks are coefficients associated with $\hat h$ and $\hat g$, respectively. The white blocks are identically zero.
Panel (a): $L_2$ is divided by the double border into $4$ submatrices $(L_2^{\vec \eps_1})_{\vec \eps_1 \in \cE_1}$, each of which has $4$~rows and $11$~columns (where only $5$ columns are nonzero). 
Panel (b): $L_2'$ is formed by vertically stacking the $3\times 11$ submatrices $(L_2^{\vec \eps_1})'$ for all $\vec \eps_1 \in \cE_1$. The old and new row numbers are printed in black and red, respectively.}
	\end{figure}

\begin{example}[$T = 2$ and $d = 2$] \label{ex: T=2 d=2}
	Note that \eqref{eq: multiperiod v} reads 
	$$V = h_{1,1} \, (X_{2,1} - X_{1,1}) + h_{1,2} \, (X_{2,2} - X_{1,2}) + g_1 + g_2.$$
	The matrix $L_2$ has $16$~rows and $11$~columns; see Figure~\ref{fig: mat}. Since $\cE_1 = \{(0,0), (0,1), (1,0), (1,1)\}$, the column ordering corresponds to the vector of $11$~variables
	\begin{equation} \label{eq: column order}
		\left[ \hat h_{1,1}(x^{\vec \eps_1}), \hat h_{1,2}(x^{\vec \eps_1})  \text{ for } \vec \eps_1 \in \cE_1 \right], \
			    \hat g_1( x_{2,1}^0),\ \hat g_1( x_{2,1}^1), \ \hat g_2( x_{2,2}^1).
	\end{equation}
To show that $L_2$ has full column rank, we proceed in four steps.\\[.6em]
	\textit{Step~1.} Our goal is to divide $L_2$ into $4$ submatrices $(L_2^{\vec \eps_1})_{\vec \eps_1 \in \cE_1}$, each of which has $4$~rows and $11$~columns. We will identify a linearly dependent row in each $L_2^{\vec \eps_1}$. 

	Consider \eqref{eq: D subset} with $T=2$. For each $\vec \eps_1 \in \cE_1$, evaluating $V$ on $D^{\vec \eps_1}$ gives rise to a $4 \times 11$ submatrix $L_2^{\vec \eps_1}$ of $L_2$. It has $4$~rows as $|D^{\vec \eps_1}| = 4$; there are in total $4$ submatrices as $|\cE_1| = 4$. The coefficients across these submatrices follow a similar pattern.
	For concreteness, consider $\vec \eps_1 = (0, 0) \in \cE_1$. Recalling the definition of price increments from~\eqref{eq: Delta def}, evaluating $V$ on $D^{\vec \eps_1}$ gives
	\begin{equation} \label{eq: submat}
	\begin{bmatrix}[1.33]
		\Delta_{2,1}^{0, 0} & \Delta_{2,2}^{0, 0} & \cdots & 1 & & \\
		\Delta_{2,1}^{0, 0} & \Delta_{2,2}^{1, 0} & \cdots & 1 & & 1\\
		\Delta_{2,1}^{1, 0} & \Delta_{2,2}^{0, 0} & \cdots & & 1 & \\
		\Delta_{2,1}^{1, 0} & \Delta_{2,2}^{1, 0} & \cdots & & 1 & 1\\
	\end{bmatrix} \begin{bmatrix}[1.25]
	h_{1,1}(x^{(0,0)}) \\ h_{1,2}(x^{(0,0)}) \\ \vdots \\ g_1(x_{2,1}^0) \\ g_1(x_{2,1}^1) \\ g_2(x_{2,2}^1)   \end{bmatrix} =
	\begin{bmatrix}[1.33]
		v(x^{(0,0)}, x_{2,1}^0, x_{2,2}^0) \\ v(x^{(0,0)}, x_{2,1}^0, x_{2,2}^1) \\ v(x^{(0,0)}, x_{2,1}^1, x_{2,2}^0) \\ v(x^{(0,0)}, x_{2,1}^1, x_{2,2}^1) \end{bmatrix},
\end{equation}
where the omitted entries (including $6$ entire columns) are zero. 
Denote the matrix in \eqref{eq: submat} by $L_2^{\vec \eps_1}$ and let $R_i$ stand for its Row $i$.
Note that $R_1 - R_2 = R_3 - R_4$.
Since $R_4$ is a linear combination of the other rows, dropping it from $L_2^{\vec \eps_1}$ does not alter the matrix rank.
This argument applies to all $\vec \eps_1 \in \cE_1$. Denote the remaining matrix by $(L_2^{\vec \eps_1})'$. 
We form the matrix $L_2'$ by vertically stacking $(L_2^{\vec \eps_1})'$ for all $\vec \eps_1 \in \cE_1$; see Figure~\ref{fig: mat2}. It follows that $L_2'$ has the same rank as $L_2$. \\[.6em]
	\textit{Step 2.} Recall that elementary row and column operations preserve the matrix rank. In this step, we aim to bring $L_2'$ to a block lower triangular matrix
\begin{equation} \label{eq: C}
	\begin{bmatrix}[1.5] L_2'' & \\ B & C \end{bmatrix}, \quad \text{where} \quad C = \begin{bmatrix}
		 1 & & \\
		 1 & & 1\\
		 & 1 &
	\end{bmatrix}
\end{equation}
	by suitable row operations. Indeed, let $\mathds 1 = (1, 1) \in \cE_1$ and $\vec \eps_1 \neq \mathds 1$ in $\cE_1$. After subtracting $( L_2^{\mathds 1} )'$ from $(L_2^{\vec \eps_1})'$, we denote the resultant matrix by $(L_2^{\vec \eps_1})''$. By applying this procedure to all $\vec \eps_1 \neq \mathds 1$, we bring $L_2'$ into the desired form, where 
\begin{equation} \label{eq: L_2''}
	L_2'' = \begin{bmatrix}[1.33]
		\Delta_{2,1}^{0, 0} & \Delta_{2,2}^{0, 0} & & & & & -\Delta_{2,1}^{0, 1} & -\Delta_{2,2}^{0, 1} \\
		\Delta_{2,1}^{0, 0} & \Delta_{2,2}^{1, 0} & & & & & -\Delta_{2,1}^{0, 1} & -\Delta_{2,2}^{1, 1} \\
		\Delta_{2,1}^{1, 0} & \Delta_{2,2}^{0, 0} & & & & & -\Delta_{2,1}^{1, 1} & -\Delta_{2,2}^{0, 1} \\
		\hline \hline
		& & \Delta_{2,1}^{0, 0} & \Delta_{2,2}^{0, 1} & & & -\Delta_{2,1}^{0, 1} & -\Delta_{2,2}^{0, 1} \\
		& & \Delta_{2,1}^{0, 0} & \Delta_{2,2}^{1, 1} & & & -\Delta_{2,1}^{0, 1} & -\Delta_{2,2}^{1, 1} \\
		& & \Delta_{2,1}^{1, 0} & \Delta_{2,2}^{0, 1} & & & -\Delta_{2,1}^{1, 1} & -\Delta_{2,2}^{0, 1} \\
		\hline \hline
		& & & & \Delta_{2,1}^{0, 1} & \Delta_{2,2}^{0, 0} & -\Delta_{2,1}^{0, 1} & -\Delta_{2,2}^{0, 1} \\
		& & & & \Delta_{2,1}^{0, 1} & \Delta_{2,2}^{1, 0} & -\Delta_{2,1}^{0, 1} & -\Delta_{2,2}^{1, 1} \\
		& & & & \Delta_{2,1}^{1, 1} & \Delta_{2,2}^{0, 0} & -\Delta_{2,1}^{1, 1} & -\Delta_{2,2}^{0, 1} \\
	\end{bmatrix}.
\end{equation}
	Clearly $C$ has full rank, and it follows that $L_2$ has full column rank if and only if $L_2''$ does.\\[.6em]
	\textit{Step 3.} Recall the definition of stock price differences~\eqref{eq: delta def}. Let $\vec \eps_1 \neq \mathds{1}$. In $(L_2^{\vec \eps_1})''$, subtracting its Row~$1$ from all other rows leaves each of them with precisely two nonzero entries of the same magnitude, $\delta_{2,j}$, with the opposite signs; cf.\ \eqref{eq: Delta diff}. Applying this procedure to all $\vec \eps_1 \neq \mathds 1$ yields
\begin{equation} \label{eq: ex step3.1}
	\begin{bmatrix}[1.33]
		\Delta_{2,1}^{0, 0} & \Delta_{2,2}^{0, 0} & & & & & -\Delta_{2,1}^{0, 1} & -\Delta_{2,2}^{0, 1} \\
		& -\delta_{2,2} & & & & & & \delta_{2,2} \\
		-\delta_{2,1} & & & & & & \delta_{2,1} &  \\
		\hline \hline
		& & \Delta_{2,1}^{0, 0} & \Delta_{2,2}^{0, 1} & & & -\Delta_{2,1}^{0, 1} & -\Delta_{2,2}^{0, 1} \\
		& & & -\delta_{2,2} & & & & \delta_{2,2} \\
		& & -\delta_{2,1} & & & & \delta_{2,1} & & \\
		\hline \hline
		& & & & \Delta_{2,1}^{0, 1} & \Delta_{2,2}^{0, 0} & -\Delta_{2,1}^{0, 1} & -\Delta_{2,2}^{0, 1} \\
		& & & & & -\delta_{2,2} & &  \delta_{2,2} \\
		& & & & -\delta_{2,1} & & \delta_{2,1} \\
	\end{bmatrix}.
\end{equation}

	Recall that the odd and even numbered columns correspond to the variables~$\hat h_{1,1}(x^{\vec \eps_1})$ and $\hat h_{1,2} (x^{\vec\eps_1})$ for $\vec \eps_1 \in \cE_1$, respectively. In particular, Column~$7$ corresponds to $\hat h_{1,1}(x^{\mathds 1})$ and Column~$8$ corresponds to $\hat h_{1,2}(x^{\mathds 1})$. We add Columns~$1,3,5$ to Column~$7$, and Columns~$2,4,6$ to Column~$8$:
\begin{equation} \label{eq: ex step3.2}
	\begin{bmatrix}[1.33]
		\Delta_{2,1}^{0, 0} & \Delta_{2,2}^{0, 0} & & & & & -\delta_{1,1} & -\delta_{1,2} \\
		& -\delta_{2,2} & & & & & & \\
		-\delta_{2,1} & & & & & & &  \\
		\hline \hline
		& & \Delta_{2,1}^{0, 0} & \Delta_{2,2}^{0, 1} & & & -\delta_{1,1} \\
		& & & -\delta_{2,2} & & & &  \\
		& & -\delta_{2,1} & & & & & & \\
		\hline \hline
		& & & & \Delta_{2,1}^{0, 1} & \Delta_{2,2}^{0, 0} & & -\delta_{1,2} \\
		& & & & & -\delta_{2,2} & &  \\
		& & & & -\delta_{2,1} & & \\
	\end{bmatrix}.
\end{equation}
All rows except Rows~$1,4,7$ have exactly one nonzero entry, and none of these entries share a common column. 
(We remark that Rows~$1,4,7$ correspond to the multi-indices $(\vec \eps_1, 0, 0)$ for $\vec\eps_1 \neq \mathds 1$ in $\cE_1$.)
Dropping these rows and columns, we denote the remaining $3 \times 2$ submatrix by $M_2$:
\begin{equation} \label{eq: M2}
	M_2=\begin{bmatrix}
		 -\delta_{1,1} & -\delta_{1,2} \\
		 -\delta_{1,1} & \\
		& -\delta_{1,2}
	\end{bmatrix}.
\end{equation}
It follows that $L_2$ has full column rank if and only if $M_2$ does.\\[.6em]
\textit{Step 4.} By inspection, $M_2$ indeed has full column rank. The example is complete.
\end{example}
\begin{remark} \label{rmk: base case}
	The example can be generalized to $T=2$ and $d \geq 1$ without much effort; cf.\ Steps 1 to 3 in the proof of Lemma \ref{lemma: rank} below. In the presence of $d$ stocks, (\ref{eq: C}) holds with a $(d+1)\times (d+1)$ matrix $C$ of binary coefficients
\begin{equation} \label{eq: C d}
	C = \begin{bmatrix}
		1 &  &  &  &  \\
		1 &  &  &  & 1 \\
    \vdots&  &  &\reflectbox{$\ddots$} \\
		1 &  &1& \\
		   &1&  &
	\end{bmatrix}.
\end{equation}
Evidently, $C$ has full rank.
In this case, $M_2$ has $(2^{d}-1)$ rows and $d$ columns; cf.\ \eqref{eq: M2}. For $1 \leq j \leq d$, the coefficients in Column $j$ alternate among
	$$
		( \ \underbrace{-\delta_{1,j}, \dots, -\delta_{1,j}}_{\times 2^{d-j}}, \underbrace{0, \dots, 0}_{\times 2^{d-j}} \ ).
	$$
It is straightforward to verify that $M_2$ has full column rank; cf.\ \eqref{eq: C''} below. (In fact, the rows in $M_2$ are linear combinations of the rows in $C''$.)
\end{remark}

\begin{proof}[Proof of Lemma \ref{lemma: rank}]
	We proceed in four steps. In Step 1, we show that keeping only a selected number of rows in $L_T$ does not decrease the matrix rank. In Steps 2 and 3, we use row and column operations to reduce the problem to showing that a certain matrix $M_T$ has full column rank. In Step 4, we establish the latter by induction on $T \geq 2$. (Steps 1 to 3 are generalizations/abstractions of Example \ref{ex: T=2 d=2}; the matrix in Figure 2 may serve as an illustration for $L_T$.)
\begin{figure}[h!]
		\centering
		\label{fig: L general}
		\includegraphics[width=0.9\textwidth, trim={1.75cm, 1.9cm, 2cm, 2cm}, clip]{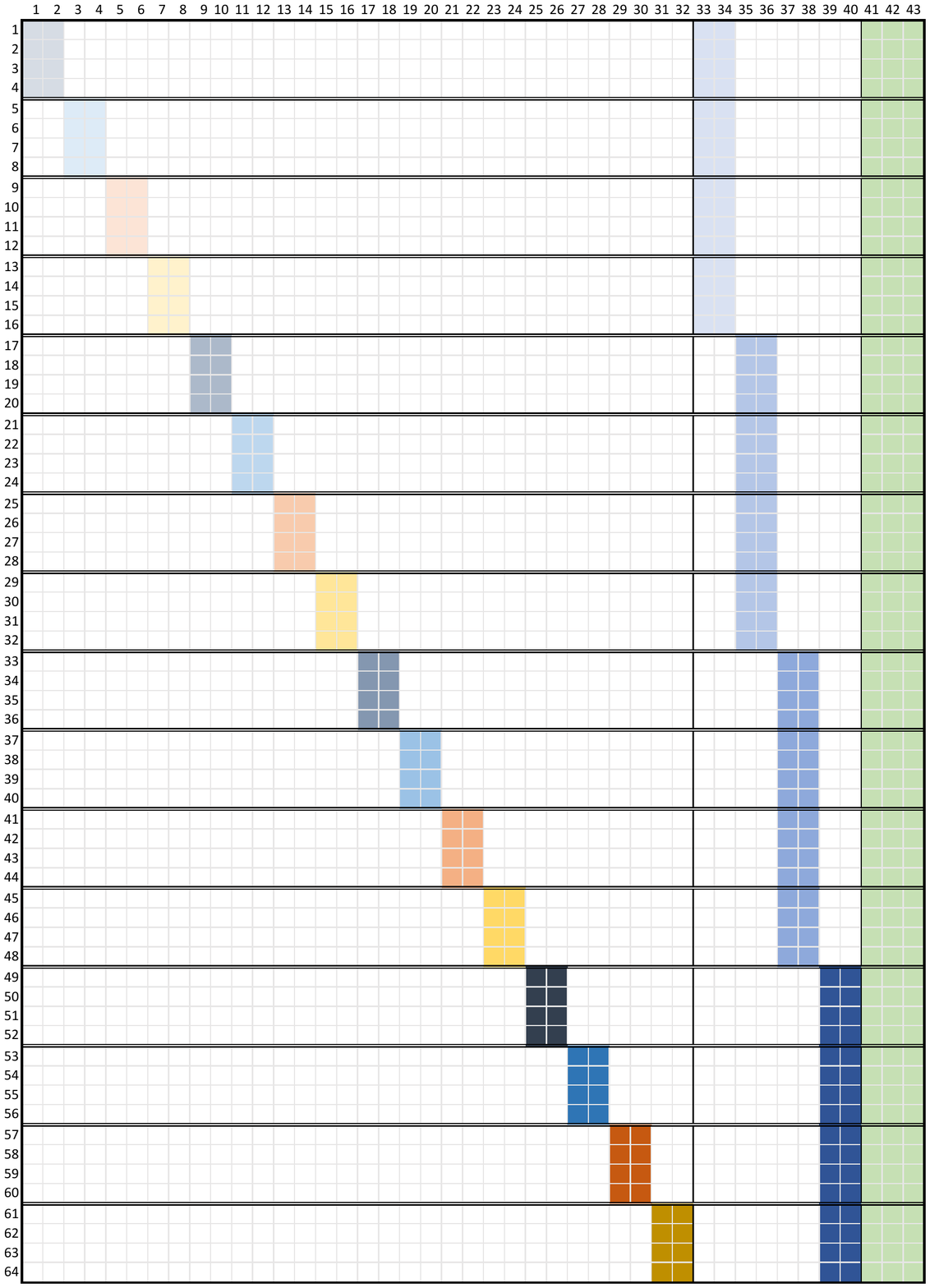}
		\caption{
		Illustration of $L_T$ in the proof of Lemma~\ref{lemma: rank} for $T=3$ and $d=2$. Columns 1--32 correspond to the variables $\hat h_{2,j}(x^{\vec\eps_2})$ for $1 \leq j \leq 2$ and $\vec \eps_2 \in \cE_2$, Columns 33--40 correspond to the variables $\hat h_{1,j}(x^{\vec\eps_1})$ for $1 \leq j \leq 2$ and $\vec \eps_1 \in \cE_1$, and lastly Columns 41--43 correspond to the variables $\hat g_{1}(x_{3,1}^0), \hat g_{1}(x_{3,1}^1), \hat g_{2}(x_{3,2}^1)$.
		}
	\end{figure}

\noindent \textit{Step 1.}  We will divide $L_T$ into $2^{\, d \, (T-1)}$ submatrices $L_T^{\vec \eps_{T-1}}$ for $\vec \eps_{T-1} \in \cE_{T-1}$, each of which has $2^d$~rows and $N_c$~columns. Then, we will identify and remove linearly dependent rows in each submatrix. Thus, the remaining matrix, denoted by $L_T'$, will have the same rank as $L_T$.

	Recall the definition \eqref{eq: D subset}. For each $\vec \eps_{T-1} \in \cE_{T-1}$, evaluating $V$ on $D^{\vec \eps_{T-1}}$ gives rise to a $2^d \times N_c$ submatrix $L_T^{\vec \eps_{T-1}}$ of $L_T$. It has $2^d$~rows as $|D^{\vec \eps_{T-1}}| = 2^d$; there are in total $2^{\, d \, (T-1)}$ submatrices as $|\cE_{T-1}| = 2^{\, d \, (T-1)}$. 

Fix $\vec \eps_{T-1} \in \cE_{T-1}$. Consider two points in $D^{\vec \eps_{T-1}}$ that share all but one entry. Formally, they can be denoted as $x^{\vec\eta_T}$ and $x^{\vec\rho_T}$ for $\vec \eta_T, \vec \rho_T \in \cE_T$ such that for some fixed $1 \leq k \leq d$, we have $\eta_{t,j} = \rho_{t,j} = \eps_{t,j}$ for all $1 \leq t \leq T-1$ and $1 \leq j \leq d$, $\eta_{T,j} = \rho_{T,j}$ for all $j \neq k$, and lastly $\eta_{T,k} = 0$ and $\rho_{T,k} = 1$.
Then, it follows from (\ref{eq: multiperiod v}) that
\begin{equation} \label{eq: diff}
	V(x^{\vec \eta_T}) - V(x^{\vec \rho_T}) = \hat h_{T-1, k}(x^{\vec \eps_{T-1}}) \, (x_{T, k}^0 - x_{T, k}^1) + \hat g_k(x_{T, k}^0) - \hat g_k(x_{T, k}^1),
\end{equation}
where the right-hand side does not depend on any $x_{T,j}^{\eta_{T,j}} = x_{T,j}^{\rho_{T,j}}$ for $j \neq k$ (i.e., the shared stock prices at date $T$).

Let $\bm{0}$ and $\bm{e}_j,\, 1 \leq j \leq d$, be the zero vector and the unit vectors in $\{0,1\}^d$. Note that $\vec \eps_{T-1} \in \cE_{T-1}$ is still fixed, and that $(\vec \eps_{T-1}, \bm 0)$ and $(\vec \eps_{t-1}, \bm e_j), \, 1 \leq j \leq d$, are $d+1$ elements in $\cE_{T}$. Let $(L_T^{\vec \eps_{T-1}})'$ be the submatrix of $L_T^{\vec \eps_{T-1}}$ whose rows are generated by these elements. By repeated applications of~\eqref{eq: diff}, it is not hard to see that all rows in $L_T^{\vec \eps_{T-1}}$ are linear combinations of the rows in $(L_T^{\vec \eps_{T-1}})'$. This argument applies to all $\vec \eps_{T-1} \in \cE_{T-1}$. We form the matrix $L_T'$ by vertically stacking $(L_T^{\vec \eps_{T-1}})'$ for all $\vec \eps_{T-1} \in \cE_{T-1}$. It follows that $L_T'$ has the same rank as $L_T$. 

As a remark for later use, we note that for each $1 \leq t \leq T-1$, coefficients of $L_T$ (and $L_T'$) in the columns corresponding to $\hat{h}_{t,j}(x^{\vec\eps_t})$ for $1\leq j\leq d$ are nonzero if and only if they belong to the rows that arise from evaluating $V$ on $D^{\vec\eps_t}$. In other words, for each $1 \leq t \leq T-1$, the $d\,2^{\,d\,t}$ columns of $L_T$ that correspond to $\hat{h}_{t,j}(x^{\vec\eps_t})$ for $1\leq j\leq d$ and $\vec\eps_t \in \cE_t$ form a block diagonal matrix, where each block has $2^{\,d\,(T-t)}$ rows and $d$ columns.\\[.6em]
\textit{Step 2.}  After some block-by-block row operations, we will bring $L_2'$ to a block lower triangular matrix
\begin{equation*} \label{eq: block mat} 
	\begin{bmatrix} L_T'' & \\ B & C \end{bmatrix}, \quad \text{with } C \text{ as in } \eqref{eq: C d}.
\end{equation*}
For $1 \leq t \leq T-1$, we denote $(1,\dots, 1) \in \cE_{t}$ by $\mathds{1}_t$. Fix $\vec \eps_{T-1} \neq \mathds 1_{T-1}$ in $\cE_{T-1}$. After subtracting $( L_T^{\mathds 1_{T-1}} )'$ from $(L_T^{\vec \eps_{T-1}})'$, we denote the resultant matrix by $(L_T^{\vec \eps_{T-1}})''$. Applying this procedure to all $ \vec \eps_{T-1} \neq \mathds 1_{T-1}$ brings $L_T'$ into the form above, where the submatrix $[\, B \;\; C \, ]$ is precisely $( L_T^{\mathds 1_{T-1}} )'$. Our choice of $d+1$ rows in Step 1 guarantees that $C$ is the matrix~\eqref{eq: C d} and has full rank. Thus, it follows that $L_T$ has full column rank if and only if $L_T''$ does. 

To understand the structure of $L_T''$, we need to recall the remark in Step~1 and the stock price difference~\eqref{eq: delta def}. Note that the columns of $L_T''$ that correspond to $\hat h_{t,j}(x^{\vec\eps_t})$ for $1 \leq t \leq T-1, \ 1 \leq j \leq d$ and $\vec \eps_t \neq \mathds 1_t$ in $\cE_t$ remain the same as in $L_T'$, since these columns are identically zero in $( L_T^{\mathds 1_{T-1}} )'$. On the other hand, the nonzero coefficients of $(L_T^{\mathds 1_{T-1}})'$ in the columns that correspond to $\hat h_{t,j}(x^{\mathds 1_t})$ are subtracted from the coefficients of $L_T'$ in these columns. We state two particular instances for later use:
\begin{enumerate}
	\item The columns in $L_T''$ that correspond to $\hat h_{T-1,j}(x^{\mathds 1_{T-1}})$ for $1\leq j\leq d$ repeat the negative values of those columns in $(L_{T}^{\mathds 1_{T-1}})'$; cf.\ \eqref{eq: L_2''}.
	\item For any $\eps_{T-1} \in \{0,1\}^{d}$ and $1 \leq j \leq d$, the coefficient of $L_T''$ in the row corresponding to $(\mathds 1_{T-2}, \eps_{T-1}, \bm{0})$ and the column corresponding to $\hat h_{T-2,j}(x^{\mathds 1_{T-2}})$ is
	\begin{equation*} %
	(x_{T-1, j}^{\eps_{T-1,j}} -x_{T-2, j}^1) - (x_{T-1, j}^1 -x_{T-2, j}^1) =
		\begin{cases} 
			\delta_{T-1, j}  & \text{if } \eps_{T-1,j} = 0 \\ 
			\ \ 0 				& \text{if } \eps_{T-1,j} = 1. 
		\end{cases}
	\end{equation*}
	Analogous patterns hold in the columns corresponding to $\hat h_{t,j}(x^{\mathds 1_{t}})$ for $1 \leq t \leq T-3$.\\
\end{enumerate}

\noindent \textit{Step 3.} Note that  $L_T''$ has $(d+1)(2^{\, d \, (T-1)} - 1)$ rows and $d \sum_{t=1}^{T-1} 2^{\, d \, t}$ columns. %
Fix $\vec \eps_{T-1} \neq \mathds 1_{T-1}$. In~$(L_T^{\vec \eps_{T-1}})''$, we subtract its Row~$1$ from all other rows. In view of~(i) in Step~2, a variant of \eqref{eq: diff} implies that for $1\leq j \leq d$, Row~$(d+2-j)$ has two nonzero entries of the same magnitude, $\delta_{T,j}$, with the opposite signs. More precisely, $-\delta_{T,j}$ is in the column corresponding to $\hat{h}_{T-1,j}(x^{\vec \eps_{T-1}})$ and $\delta_{T,j}$ is in the column corresponding to $\hat h_{T-1, j}(x^{\mathds 1_{T-1}})$; cf.\ \eqref{eq: ex step3.1}.

For $1 \leq j \leq d$, we add the column corresponding to $\hat{h}_{T-1,j}(x^{\vec \eps_{T-1}})$  to the column corresponding to $\hat h_{T-1,j}(x^{\mathds 1_{T-1}})$. As a result, for $1 \leq j \leq d$, Row~$(d+2-j)$ now has precisely one nonzero entry, $-\delta_{T,j}$, and this entry is located in the column corresponding to $\hat{h}_{T-1,j}(x^{\vec \eps_{T-1}})$; cf.~\eqref{eq: ex step3.2}. 
Moreover, for $1 \leq j \leq d$, the first coefficient of $(L_T^{\vec \eps_{T-1}})''$ in the column corresponding to $\hat{h}_{T-1,j}(x^{\mathds 1_{T-1}})$ is
\begin{equation} \label{eq: first d columns}
	\begin{cases} -\delta_{T-1, j} & \text{if } \eps_{T-1,j} = 0 \\ \ \ \ 0 & \text{if } \eps_{T-1,j} = 1 \end{cases}
\end{equation}
and all other coefficients are zero. We apply this procedure to all $\vec \eps_{T-1} \neq \mathds 1_{T-1}$. As a result, $d(2^{\, d \, (T-1)} - 1)$ rows have only one nonzero entry  each, and these entries respectively belong to the first $d(2^{\, d \, (T-1)} - 1)$ columns of $L_T''$. Dropping these rows and columns, we denote the remaining matrix by $M_T$. Thus, we have shown that $L_T$ has full column rank if and only if $M_T$ does. \\[.6em] 
\textit{Step 4.} We now show that $M_T$ indeed has full column rank. We argue by induction on $T$; the base case $T=2$ was established in Remark \ref{rmk: base case}. Note that $M_T$ has $(2^{\, d \, (T-1)} - 1)$ rows and $d \sum_{t=0}^{T-2} 2^{\, d \, t}$ columns, where 
\begin{itemize}
	\item the rows correspond to $(\vec\eps_{T-1}, \bm{0})$ for all $\vec\eps_{T-1} \neq \mathds 1_{T-1}$ in $\cE_{T-1}$;
	\item the first $d$ columns correspond to $\hat h_{T-1, j}(x^{\mathds 1_{T-1}})$ for $1\leq j \leq d$; and
	\item the rest of the columns correspond to $\hat h_{t,j}(x^{\vec \eps_t})$ for $t \in \{T-2, \dots, 1\}$, $\vec\eps_t \in \cE_t$ and $j \in \{1,\dots, d\}$.
\end{itemize}
The matrix in Figure~\ref{fig: mat3} may serve as an illustration. As
$$
	\left\{ (\vec\eps_{T-1}, \bm{0}) :  \vec\eps_{T-1} \in \cE_{T-1} \right\} = \big\{ (\vec\eps_{T-2}, \eps_{T-1}, \bm{0}) : \vec\eps_{T-2} \in \cE_{T-2}, \ \eps_{T-1} \in \{0,1\}^d \big\}
$$
we can divide~$M_T$ into submatrices~$(M_T^{\vec\eps_{T-2}})_{\vec\eps_{T-2}\in\cE_{T-2}}$ such that the rows of~$M_T^{\vec\eps_{T-2}}$ correspond to $(\vec\eps_{T-2}, \eps_{T-1}, \bm{0})$ for all $\eps_{T-1} \in \{0,1\}^d$. Each~$M_T^{\vec\eps_{T-2}}$ has $2^d$ rows, with the exception that $M_T^{\mathds 1_{T-2}}$ has $2^d-1$ rows (since $M_T$ does not contain the row corresponding to $\mathds 1_{T-1}$). For this reason, we need to treat the cases $d\geq2$ and $d=1$ separately.\\[.6em]

\begin{figure}[b!]
		\centering
		\subfigure[$M_3$]{
			\label{fig: mat3}
			\includegraphics[width=0.45\textwidth, trim={1.5cm, 17cm, 12cm, 2cm}, clip]{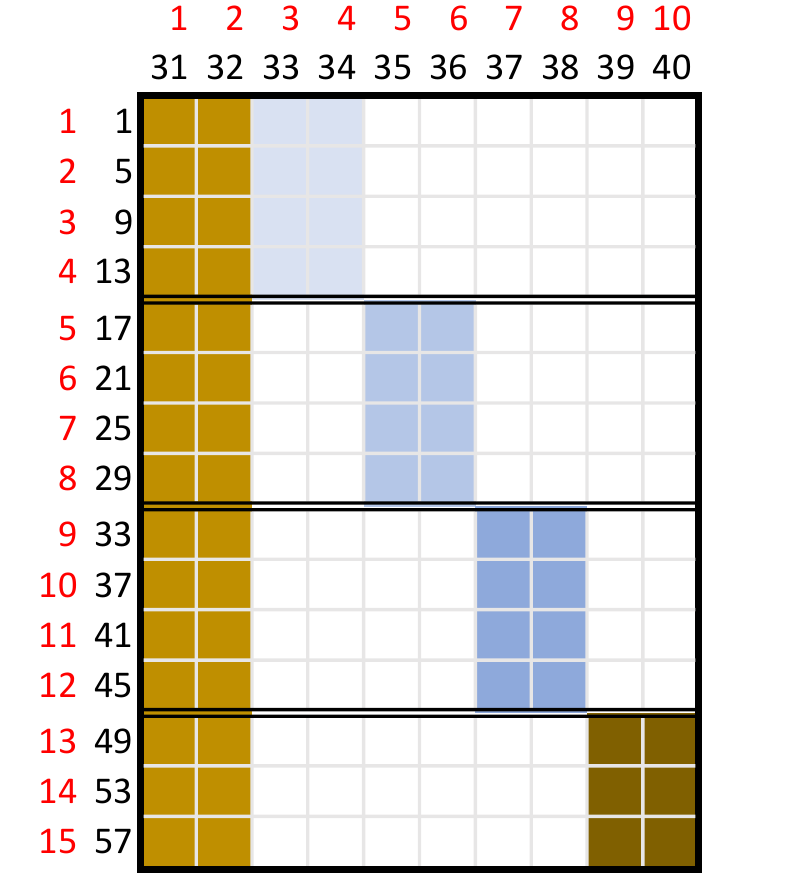}}
		\qquad
		\subfigure[$M_3'$]{
			\label{fig: mat4}
			\includegraphics[width=0.45\textwidth, trim={1.5cm, 17cm, 12cm, 2cm}, clip]{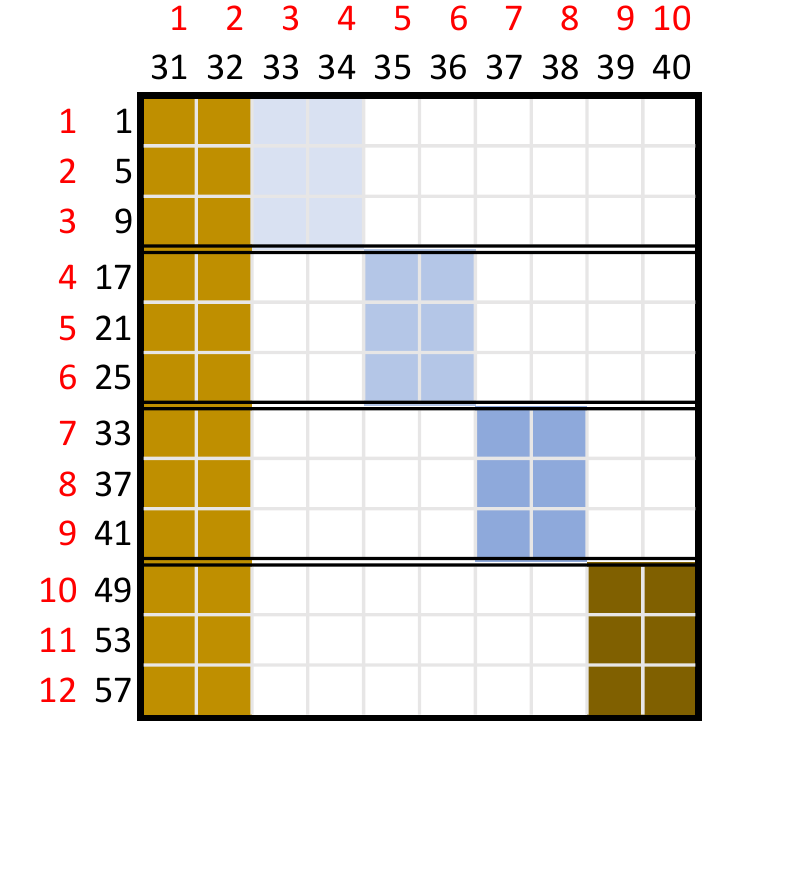}}
		\caption{
		Illustration of Step~4 in the proof of Lemma \ref{lemma: rank} for $T=3$ and $d=2$. The row and column numbers in black are inherited from $L_3$, and the new row and column numbers are printed in red. The first two columns correspond to $\hat h_{2,j}(x^{\mathds 1_2})$ for $j=1,2$; the rest of the columns correspond to $\hat h_{1,j}(x^{\vec\eps_1})$ for $j=1,2$ and $\vec\eps_1 \in \cE_1$. Coefficients in each light gold block are the negatives of those in the dark gold block.}
	\end{figure}
(a) Case $d \geq 2$: Similarly to Step 1, the rows in $M_T$ are linearly dependent, and we can reduce $M_{T}$ to a matrix $M_T'$ of $(d+1)2^{\, d \, (T-2)}$~rows without decreasing its rank.
That is, $M_T'$ consists of $2^{\, d \, (T-2)}$ submatrices $(M_T^{\vec \eps_{T-2}})'$ of $d+1$ rows for $\vec \eps_{T-2} \in \cE_{T-2}$. As a consequence of \eqref{eq: first d columns}, the first $d$ columns of $(M_T^{\vec\eps_{T-2}})'$ make up a $(d+1) \times d$ matrix
\begin{equation} \label{eq: C''}
	C'' := \begin{bmatrix}
		-\delta_{T-1,1} & -\delta_{T-1,2} & & \dots & &-\delta_{T-1,d-1} & -\delta_{T-1,d} \\
		\vdots & \vdots & & \dots & &-\delta_{T-1,d-1} & 0 \\
		\vdots & \vdots & &\dots  & &0 & -\delta_{T-1,d} \\
	      \vdots & \vdots  & & &\reflectbox{$\ddots$} & -\delta_{T-1,d-1} & \vdots \\
		 \vdots  & \vdots  & & \reflectbox{$\ddots$} & &\vdots & \vdots \\
		\vdots & -\delta_{T-1,2}&\reflectbox{$\ddots$} & \dots & & \vdots & \vdots \\
		-\delta_{T-1,1} & 0 & & \dots & & \vdots & \vdots \\
		0 & -\delta_{T-1,2} & & \dots & & -\delta_{T-1,d-1} & -\delta_{T-1,d}
	\end{bmatrix}
\end{equation}
which does not depend on $\vec\eps_{T-2}$ and has full column rank. Thus, the first $d$ columns of $M_T'$ are formed by vertically stacking $C''$ $2^{\, d \, (T-2)}$ times.

Recalling~(ii) from Step~2, we note that the last $d+1$ rows of $M_T'$ correspond to $(\mathds 1_{T-2}, \bm{0}, \bm{0})$ and $(\mathds 1_{T-2}, \bm{e}_j, \bm{0})$ for $1 \leq j \leq d$. As a result, the submatrix of $M_T'$ formed by taking the last $d+1$ rows and the columns corresponding to $h_{T-2, j}(x^{\mathds 1_{T-2}})$ for $1\leq j \leq d$ is precisely $-C''$. Now, adding the columns corresponding to $\hat h_{T-1, j}(x^{\mathds 1_{T-1}})$ to the columns corresponding to $ \hat h_{T-2, j}(x^{\mathds 1_{T-2}})$ for $1 \leq j \leq d$ can bring $M_T'$ to the block triangular form
\begin{equation*} \label{eq: block mat}
	\begin{bmatrix} A'' & M_T'' \\ C'' & \end{bmatrix}.
\end{equation*}
\\[.6em] 
(b) Case  $d = 1$: There is no need to remove any rows from $M_T$. We set $M_T' = M_T$ and note that it has $2^{T-1} - 1$ rows. A subtle difference to (a) is that $M_T'$ does not contain the row corresponding to $(\mathds 1_{T-1}, \bm{e}_1, \bm{0})$. (We remark that $\bm{e}_1 = 1$ and $\bm{0} = 0$ in the case $d=1$.) The coefficients of $M_T'$ in the first column alternate between $-\delta_{T-1,1}$ and $0$. Also, the last row of $M_T'$ has only one nonzero entry, $\delta_{T-1,1}$, in the column corresponding to $h_{T-2, 1}(x^{\mathds 1_{T-2}})$. Adding the column corresponding to $\hat h_{T-1, 1}(x^{\mathds 1_{T-1}})$ to the column corresponding to $ \hat h_{T-2, 1}(x^{\mathds 1_{T-2}})$ brings $M_T'$ to the form above with the $1\times 1$ matrix $C'' = [-\delta_{T-1,1}]$. \\[.6em] 
Therefore, in either case, $M_{T}''$ is a matrix with $(d+1)(2^{\, d \, (T-2)}-1)$ rows and $d \sum_{t=1}^{T-2} 2^{\, d \, t}$ columns. It remains to show that it has full column rank.
Applying the same procedure as in Step 3 to $M_T''$, we see that $d(2^{\, d \, (T-2)} - 1)$ rows have only one nonzero entry each, and none of these entries share a common column. We then drop these rows and columns. One can check that the remaining matrix, of $(2^{\, d \, (T-2)}-1)$ rows and $d \sum_{t=0}^{T-3} 2^{\, d \, t}$ columns, is precisely $M_{T-1}$; we omit further details in the interest of brevity. The inductive hypothesis applies.
\end{proof}

\bibliography{stochfin_LZ}

\newcommand{\dummy}[1]{}
\begin{thebibliography}{10}

\bibitem{AcciaioLarssonSchachermayer.17}
B.~Acciaio, M.~Larsson, and W.~Schachermayer.
\newblock The space of outcomes of semi-static trading strategies need not be
  closed.
\newblock {\em Finance Stoch.}, 21(3):741--751, 2017.

\bibitem{BeiglbockHenryLaborderePenkner.11}
M.~Beiglb{\"o}ck, P.~Henry-Labord{\`e}re, and F.~Penkner.
\newblock Model-independent bounds for option prices: a mass transport
  approach.
\newblock {\em Finance Stoch.}, 17(3):477--501, 2013.

\bibitem{BeiglbockNutzTouzi.15}
M.~Beiglb{\"o}ck, M.~Nutz, and N.~Touzi.
\newblock Complete duality for martingale optimal transport on the line.
\newblock {\em Ann. Probab.}, 45(5):3038--3074, 2017.

\bibitem{BorweinLewis.92}
J.~M. Borwein and A.~S. Lewis.
\newblock Decomposition of multivariate functions.
\newblock {\em Canad. J. Math.}, 44(3):463--482, 1992.

\bibitem{BouchardNutz.11}
B.~Bouchard and M.~Nutz.
\newblock Weak dynamic programming for generalized state constraints.
\newblock {\em SIAM J. Control Optim.}, 50(6):3344--3373, 2012.

\bibitem{BreedenLitzenberger.78}
D.~T. Breeden and R.~H. Litzenberger.
\newblock Prices of state-contingent claims implicit in option prices.
\newblock {\em J. Bus.}, 51(4):621--651, 1978.

\bibitem{DelbaenSchachermayer.06}
F.~Delbaen and W.~Schachermayer.
\newblock {\em The Mathematics of Arbitrage}.
\newblock Springer, Berlin, 2006.

\bibitem{FollmerSchied.11}
H.~F{\"{o}}llmer and A.~Schied.
\newblock {\em Stochastic Finance: An Introduction in Discrete Time}.
\newblock W. de Gruyter, Berlin, 3rd edition, 2011.

\bibitem{GalichonHenryLabordereTouzi.11}
A.~Galichon, P.~Henry-Labord{\`e}re, and N.~Touzi.
\newblock A stochastic control approach to no-arbitrage bounds given marginals,
  with an application to lookback options.
\newblock {\em Ann. Appl. Probab.}, 24(1):312--336, 2014.

\bibitem{GuasoniRasonyiSchachermayer.08}
P.~Guasoni, M.~R\'{a}sonyi, and W.~Schachermayer.
\newblock Consistent price systems and face-lifting pricing under transaction
  costs.
\newblock {\em Ann. Appl. Probab.}, 18(2):491--520, 2008.

\bibitem{Guyon.20}
J.~Guyon.
\newblock The joint {S\&P} 500/{VIX} smile calibration puzzle solved.
\newblock {\em Risk}, 2020.

\bibitem{Guyon.21}
J.~Guyon.
\newblock Dispersion-constrained martingale {S}chr{\"o}dinger problems and the
  exact joint {S\&P} 500/{VIX} smile calibration puzzle.
\newblock {\em Preprint SSRN:3853237}, 2021.

\bibitem{HenryLabordere.19}
P.~Henry-Labord{\`e}re.
\newblock From (martingale) {S}chr{\"o}dinger bridges to a new class of
  stochastic volatility model.
\newblock {\em Preprint SSRN:3353270}, 2019.

\bibitem{Hobson.98}
D.~Hobson.
\newblock Robust hedging of the lookback option.
\newblock {\em Finance Stoch.}, 2(4):329--347, 1998.

\bibitem{Hobson.11}
D.~Hobson.
\newblock The {S}korokhod embedding problem and model-independent bounds for
  option prices.
\newblock In {\em Paris-{P}rinceton {L}ectures on {M}athematical {F}inance
  2010}, volume 2003 of {\em Lecture Notes in Math.}, pages 267--318. Springer,
  Berlin, 2011.

\bibitem{Leonard.14}
C.~L\'{e}onard.
\newblock A survey of the {S}chr\"{o}dinger problem and some of its connections
  with optimal transport.
\newblock {\em Discrete Contin. Dyn. Syst.}, 34(4):1533--1574, 2014.

\bibitem{Nutz.20}
M.~Nutz.
\newblock {\em Introduction to Entropic Optimal Transport}.
\newblock Lecture notes, Columbia University, 2021.
\newblock
  \url{https://www.math.columbia.edu/~mnutz/docs/EOT_lecture_notes.pdf}.

\bibitem{NutzWieselZhao.22b}
M.~Nutz, J.~Wiesel, and L.~Zhao.
\newblock Martingale {S}chr{\"o}dinger bridges and optimal semistatic
  portfolios.
\newblock {\em Arxiv preprint}, 2022.

\bibitem{RuschendorfThomsen.97}
L.~R\"{u}schendorf and W.~Thomsen.
\newblock Closedness of sum spaces and the generalized {S}chr\"{o}dinger
  problem.
\newblock {\em Teor. Veroyatnost. i Primenen.}, 42(3):576--590, 1997.

\end{thebibliography}
\bibliographystyle{abbrv}

\end{document}